\documentclass[aip, jmp, preprint]{revtex4-1}

\usepackage[english]{babel}
\usepackage[utf8x]{inputenc}
\usepackage[T1]{fontenc}
\usepackage{xfrac}
\usepackage[toc,page]{appendix}
\usepackage[export]{adjustbox}

\usepackage{amsmath}
\usepackage{amsthm}
\usepackage{amssymb}
\usepackage{graphicx}
\usepackage{enumerate}
\usepackage[colorinlistoftodos]{todonotes}
\usepackage[colorlinks=true, allcolors=blue]{hyperref}

\usepackage{hyperref}

\usepackage{xpatch}
\makeatletter
\AtBeginDocument{\xpatchcmd{\@thm}{\thm@headpunct{.}}{\thm@headpunct{}}{}{}}

\DeclareMathOperator\supp{supp}
\makeatother

\theoremstyle{definition}

\newtheorem{exam}{Example}

\newtheorem{prop}[exam]{Proposition}
\newtheorem{cor}[exam]{Corollary}
\newtheorem{remark}[exam]{Remark}

\begin{document}

\title{Stochastic model for barrier crossings and fluctuations in local timescale.}
\author{Rajeev Bhaskaran}
\affiliation{Statistics and Mathematics Unit, Indian Statistical Institute, 8th Mile Mysore Road, Bangalore 560059, India}
\author{Vijay Ganesh Sadhasivam}
\affiliation{Statistics and Mathematics Unit, Indian Statistical Institute, 8th Mile Mysore Road, Bangalore 560059, India}
\affiliation{Yusuf Hamied Department of Chemistry, University of Cambridge, Lensfield Road, Cambridge,
CB2 1EW, UK}
\date{March 2022}

\begin{abstract}
The problem of computing the rate of diffusion-aided activated barrier crossings between metastable states is one of broad relevance in physical sciences. The \textit{transition path formalism} aims to compute the rate of these events by analysing the statistical properties of the transition path between the two metastable regions concerned. In this paper, we show that the transition path process is a unique solution to an associated stochastic differential equation (SDE), with a discontinuous and singular drift term. The singularity arises from a \textit{local time} contribution, which accounts for the fluctuations at the boundaries of the metastable regions. The presence of fluctuations at the local time scale calls for an \textit{excursion theoretic} consideration of barrier crossing events. We show that the rate of such events, as computed from excursion theory, factorizes into a local time term and an excursion measure term, which bears empirical similarity to the \textit{transition state theory} rate expression. Since excursion theory makes no assumption about the presence of a transition state in the potential energy landscape, the mathematical structure underlying this factorization ought to be general. We hence expect excursion theory (and local times) to provide some physical and mathematical insights in generic barrier crossing problems.  

\end{abstract}
\maketitle

\section{Introduction}
A systematic study of transition between \textit{metastable} states in physical systems has been of relevance in computing the rates of chemical reactions\cite{pechukas1981transition,truhlar1983current}, protein folding\cite{onuchic1997theory}, crystal nucleation\cite{Oxtoby_1992} etc. An empirical expression for the (temperature-dependent) rate of such events was first provided by Arrhenius\cite{Arrhenius+1889+96+116}, which identified the presence of an `activation energy' barrier that needs to be crossed for a transition to take place. The development of potential energy surfaces(PES) \cite{EyringPolanyi+2013+1221+1246} in the early $20^{\text{th}}$ century allowed the interpretation of chemical reactions as a continuous (albeit rare) progression from `reactants' to `products', mediated by a \textit{transition state} which is a saddle point in the PES corresponding to the aforementioned activation energy barrier. \textit{Transition state theory} (TST)\cite{eyring1935activated}, postulated shortly after this development, provided a derivation for Arrhenius' expression, under the assumption of quasiequilibrium between the reactants and transition state. While the language of transition state theory (and much of this article) is that of chemistry, its underlying structure and usefulness makes it relevant for understanding \textit{rare events} in numerous physical contexts. 

Although transition state theory serves as a strong theoretical tool in interpreting the rates of gas-phase molecular reactions, it is, by virtue of its founding assumptions, incapable of describing reaction rates where the particle dynamics is diffusive, which is the case in many condensed phase systems. To address this problem, Kramers reconceptualized chemical reactions as Brownian motion aided barrier escape from metastable states, driven by thermal fluctuations\cite{kramers1940brownian}. He was able to derive a formally exact expression for reaction rates in the limits of high and low friction\cite{hanggi1990reaction,kramers1940brownian}, which bore a clear structural resemblance to the TST rate expression, indicating that the basic empirical assumptions of activated barrier crossings are still valid. 

While the mathematical assumptions underlying Kramers theory allow for the expansion of the scope of TST and help interpret reaction rates in condensed phase systems, it hinges heavily on the presence of a  transition state in the PES. However, when a transition state cannot be unambiguously identified (for instance, when the PES terrain is `rugged', as in the protein folding problem\cite{onuchic1997theory}), the assumption and limits that Kramers considers do not hold anymore. 

Nonetheless, modeling chemical reaction as a diffusion over a PES, that is intrinsic to Kramers problems (and later extended by Grote and Hynes to non-Markovian processes in \citenum{northrup1980stable, grote1980stable}), brings the tools in the theory of stochastic differential equations (SDE) to the discussion of reaction rates. Computing reaction rates in such a mathematical setup can be achieved by adopting an alternative characterization of chemical reactions involving the \textit{transition path}, i.e. the path along which a reactive transition happens by barrier crossing. Such a viewpoint reduces the emphasis on transition states in computing reaction rates and has been described and developed as \textit{transition path sampling} (TPS)\cite{pratt1986statistical,bolhuis2002transition,bolhuis2003transition}.   

Transition path sampling primarily encompasses algorithmic procedures that attempt to sample all possible routes from reactants to products and weight them based on their relative probability of realization. Extracting useful and relevant information about the mechanism of a reaction, presence of transition states or other dynamical bottlenecks (if any) require a concrete mathematical framework, the importance of which was underlined in the development of \textit{transition path theory}(TPT) by Vanden-Eijnden et al\cite{vanden2006towards,vanden2010transition}. 

The cornerstone of TPT is the definition of \textit{reactive trajectories} from the transition path between the reactants and products, and characterization of its statistical properties. These reactive trajectories can be understood as crossings of the diffusion process (as in Kramers theory) from the reactant region (say $A$) to the product region (say $B$) through a transition region. The \textit{transition path process} is then the (discontinuous) semi-martingale associated with the crossings of the transition region by the diffusion.

The semimartingale associated with such crossings of an interval $(a,b)$ by a continuous semimartingale was systematically studied in \citenum{rajeev1990semi}. In this paper, using the theory introduced in \citenum{rajeev1990semi}, we characterize the transition path process as the unique solution of a singular SDE which is adapted to the original diffusion process. We appropriately identify an open interval $(a,b)$ as the transition region according to the specification of the PES. The rate of the reaction $\kappa$ is obtained from the number of \textit{reactive trajectories} ($N_t$) until time $t$, which, in our context will be the number of \textit{upcrossings} of $(a,b)$ by the diffusion until time $t$.

Each such upcrossing corresponds to a crossing of the activation energy barrier and is hence by preceded by long-time fluctuations around the boundary of the reactant region. A measure of the amount of such fluctuations around a point $a$ until time $t$ is provided in semimartingale theory by the \textit{local time} $L(t,a)$ at that point which is increasing in $t$ and is defined via the Tanaka formula. Every such fluctuation that happens during $[0,t]$ corresponds to an excursion of the diffusion from the point $a$ with a certain height and duration. 

Of interest in this paper are the excursions of the diffusion process from the reactant to the product region. These excursions correspond precisely to the reactive trajectories and hence can be used to compute the reaction rate $\kappa$, which is discussed in later sections of this paper. This procedure factorizes the reaction rate expression into a term involving local time fluctuations and another involving the excursion measure. It also gives us some insight into the dependence of the excursion measure on the parameters of the diffusion.

We also use classical results from renewal theory to compute the reaction rate; we first show that the diffusion process describing chemical reactions is \textit{regenerative} and \textit{positive recurrent} and then later compute $\kappa$ in terms of the expected cycle time associated with a regenerative process (compare with Kramers \textit{mean first passage time} \cite{hanggi1990reaction}). By comparing the $\kappa$ thus obtained using the two methods above, we are able to show that the factorization of reaction rate by excursion theory corresponds precisely to the Arrhenius/TST rate expression. Thus, transitions between two metastable states can always be interpreted mathematically in terms of the local time fluctuations and excursion measure. 

For the sake of clarity and to illustrate the basic idea in reaction-rate theory, we describe the excursions of the diffusion on a simple one-dimensional PES. Much of our results extend to a more general class of one dimensional diffusions that are positive recurrent.  

The paper is organized as follows: In section \ref{Prelims}, we start with the description of the reactive diffusion process and define the relevant notions required to understand the local time scale intrinsic in chemical reactions. In section \ref{SSDEsection}, we formulate a singular SDE that describes the transition path process and characterize its solutions. In section \ref{excurthsect}, we provide a primer on excursion theory and indicate its relevance in reaction rate theory. In section \ref{reactionratesect}, we compute the rate expression using excursion theory. In \ref{renewalthsect}, we compare the known rate expression computed using the renewal theorem and compare it with the one obtained from excursion theory. We conclude with a discussion on the importance and relevance of local time scales in reaction rate theory.

\section{Preliminaries}\label{Prelims}


As in typical in discussions of transition path formalism \cite{vanden2010transition}, we characterize a chemical reaction by the trajectory of a particle diffusing on a potential energy surface $U$ using a real valued stochastic process $X: [0,\infty) \times \Omega \rightarrow \mathbb{R}$ defined on a probability space $(\Omega,\mathcal{F},\text{P})$ that follows the Ito stochastic diffusion equation (SDE):
\begin{equation}\label{SDE}
    dX_t = \sigma(X_t) \; dB_t + b(X_t)dt,  \;\;\;\;\; X_0 = x_0
\end{equation}
where the process $B_t$ is a standard Brownian motion process on a measurable space $\Omega$, $\sigma$ is the \textit{diffusion coefficient}, $b(\cdot)$ is the \textit{drift} term, and $\mathcal{F}$ is the \textit{filtration} generated by the process $B_t$ satisfying the usual conditions as in \citenum{karatzas2012brownian}. Here and below, following the typical convention, we shall drop the explicit dependence on the random trajectory $\omega \in \Omega$ in our notation. 

We assume that the SDE in \eqref{SDE} describes overdamped Langevin dynamics \cite{hanggi1990reaction} and accordingly, by the fluctuation-dissipation theorem and Einstein's relations \cite{zwanzig2001nonequilibrium}, $\sigma$ is constant with $\sigma = \sqrt{2 D}$, where $D$ is the diffusion coefficient. The drift term $b(x)$ is obtained from the PES as:  
\begin{align}
    b(x) = -\frac{D}{k_B T} \frac{dU}{dx}
\end{align}
where $T$ is the temperature and $k_B$ is the Boltzmann constant. For our discussion, the PES shall be the truncated double well potential:
\begin{align}\label{PES}
    U(x) =    
\begin{cases}
 & Q(x) \hspace{3.1cm}  |x |\leq K-\epsilon \\
 &g(|x|) \hspace{3cm} |x| \in (K-\epsilon,K+\epsilon)\\
 &k(|x|-K) + Q(K)  \hspace{0.92cm} |x| \geq K+\epsilon 
\end{cases}
\end{align}

In \eqref{PES}, $Q(x)$ is the quartic double well potential $Q(x) = h^2 x^4 - \frac{1}{2}\lambda^2 x^2$,  $k \;\text{and}\; K $ are positive constants with $K \geq \left |\frac{\lambda}{h} \right |$, and $g(x)$ is a smooth function which makes $U(x)$ differentiable at $|x| = K\pm \epsilon$, so that $U \in C^1(\mathbb{R})$ with a bounded derivative. The proof of existence of such a function is given in the supporting information. A schematic plot of the PES is provided in figure \ref{fig:my_label}.

\begin{figure}
    \centering
    \includegraphics[scale=0.75]{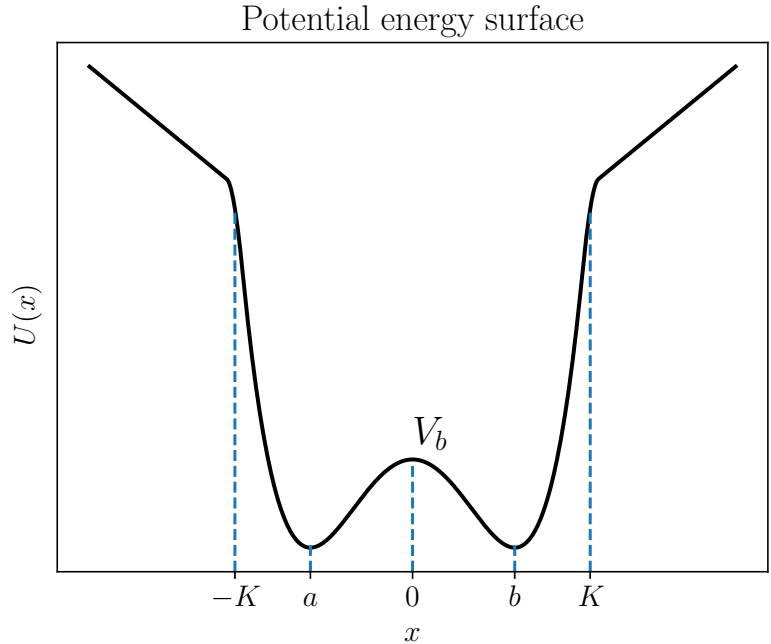}
    \caption{Illustration of the truncated double well potential $U(x)$. The points $a,b$ and the barrier height $V_b$ are as defined in the text. The quartic and the linear of the PES are connected by the function $g$ which is defined in an $\epsilon$ neighborhood around $K$, that makes $U \in C^1(\mathbb{R})$. }
    \label{fig:my_label}
\end{figure}

The double well potential is a straightforward choice for describing transitions between metastable states which are ubiquitous in chemistry and biophysics \cite{onuchic1997theory}. For a more realistic representation of such transitions, the double well can be made asymmetric by the addition of a cubic term to $Q(x)$ in \eqref{PES} and the results of this paper will still hold. The reason for truncating this double well potential at $K$ and appending it with a linear term is purely mathematical - Such a modification makes the drift term $b(x)$ Lipschitz continuous in the coordinate variable and also satisfy:
\begin{align}\label{bestimate}
    \int_0^x b(y) dy \leq m |x| \hspace{1cm} \text{for some} \; m < 0, |x| \geq x_0
\end{align}

The first condition allows \eqref{SDE} to have a unique strong solution $(X_t)$ relative to the Brownian motion $(B_t)$ that is a continuous semimartingale\cite{karatzas2012brownian} adapted to the filtration $\mathcal{F}$ (Theorem 2.5 in \citenum{karatzas2012brownian}) and the second condition allows $(X_t)$ to have a finite \textit{speed measure}, which shall be discussed later (Note that the point of truncation is taken far from the minima of $Q(x)$ so as to ensure that the linear term in $U(x)$ has negligible impact on the dynamics of $(X_t)$). This stochastic process $(X_t)$ shall be the concern of this paper. We begin first with a few definitions required for this discussion:

Let $T_x := \inf \{s : X_s = x\}$. For $y \in \mathbb R$, we denote
by $P_y$ the law of $X_t$ on $C([0,\infty))$ such that $P_y\{X_0 = y\} =1$. Note that $X_t$ is the \textit{coordinate mapping process} \cite{karatzas2012brownian} on $C([0,\infty))$, that is $X_t(\omega) = \omega(t)$ for $\omega \in C([0,\infty))$. $E_y$
denotes the expectation under $P_y$. The process $X_t$ is \textit{positive recurrent} if $\forall x,y \in \mathbb{R}$, $E_y T_x < \infty$.

Let $\tau_0 := 0 := \sigma_0$ and
define the sequence $\{\sigma_n; n \geq 1\}, \{\tau_n; n \geq 1\}$
as follows:
\begin{align*}
\sigma_n := \inf \{s > \tau_{n-1}: X_s = b\} ~{\rm and}~ \tau_n :=
 \inf \{s > \sigma_{n}: X_s = a\}
\end{align*}
Here we choose $-a = b = \left | \frac{\lambda}{2h} \right |$, the points of global minima of the quartic potential $Q(x)$ (Note that we use $b(\cdot)$ for denoting the drift term). In the context of reaction rate theory and metastable states, these points can be interpreted as the equilibrium structure of the `reactant' and `product' metastable states.

The process $(X_t)$ is said to be \textit{regenerative} at a point $x$, if $P(\theta_{\tau}X_{\cdot} \in A | \mathcal{F}_{\tau}) = P_x(A) \; a.s.$ on $\{ \tau<\infty, \; X_{\tau} = x \}$ for every \textit{stopping time} $\tau$ \cite{kallenberg1997foundations}, where $\theta_{\tau}X_{\cdot} := X_{\cdot + \tau}$. 

\begin{prop}\label{posrecur}
The unique strong solution $(X_t)$ to \eqref{SDE} is positive recurrent and regenerative at the points $a$ and $b$.
\end{prop}
\begin{proof}
The fact that $(X_t)$ is
regenerative at $a$ follows from a standard application of the strong Markov property
of the diffusion $(X_t)$ under $P_a$, which can be stated as follows \cite{karatzas2012brownian, kallenberg1997foundations}: 
\begin{align}\label{strongMarkov}
    E_a \{F(X_{\tau +
\cdot})|{\cal F}_{\tau}\} ~=~ E_{X_\tau} F(X_{\cdot}), {\rm ~a.s.~} P_a.
\end{align} 
where $F : C[0,\infty) \rightarrow \mathbb R$ is a bounded Borel-measurable function, $\tau < \infty \; a.s.$ is any stopping time and $X_{\tau + \cdot}$ denotes the path $s \rightarrow
X_{\tau +s}$. The regenerative property follows by taking $F = I_A$, where $A \subseteq C([0,\infty))$ is a Borel set and $\tau$ such that $X_{\tau} = a$.

We first define the \textit{speed measure} and the \textit{scale function}. Given the diffusion process $(X_t)$ as in \eqref{SDE}, the scale function $p$ is defined w.r.t a point $c \in \mathbb{R}$  as \cite{karatzas2012brownian}:
\begin{align}\label{scalefn}
  p(x) := \int_c^x \exp \left \{ -2 \int_c^y \frac{b(y)}{\sigma^2(y)}dy \right\} dx
\end{align}
The scale function for $(X_t)$ is well-defined since the integrand in the exponent is continuous. The speed measure\cite{karatzas2012brownian} is defined as:
\begin{align}\label{speedmeas}
    m(dx) := \frac{2\;dx}{p'(x)\sigma^2(x)}
\end{align}
The positive recurrence of the diffusion $(X_t)$ is proved as in \citenum{karatzas2012brownian} (Ex. 5.40 of Chapter 5). For completeness, we give a few details. For $x<y$, we have positive recurrence if the RHS limit in
\begin{align*}
    E_y T_x = \lim\limits_{n \rightarrow \infty} E_y(T_x \wedge T_n)
\end{align*}
is finite. From eq. (5.55) and (5.59) in \citenum{karatzas2012brownian}, 
\begin{align}
E_y(T_x \wedge T_n) = - \int_x^y(p(y)- p(z))m(dz) +
\frac{p(y)-p(x)}{p(n)-p(x)}\int_x^n (p(n)-p(z))m(dz) 
\end{align}
In the limit $n\rightarrow\infty$,
\begin{align*}
\lim\limits_{n \rightarrow \infty}\frac{1}{p(n)-p(x)}\int_x^n (p(n)-p(z))m(dz) =& m([x,\infty)) + \lim\limits_{n \rightarrow \infty}\frac{p(x)m(x,n)}{p(n)-p(x)} \\& - \lim\limits_{n
\rightarrow \infty}\frac{1}{p(n)-p(x)}\int_x^n p(z)m(dz) 
\end{align*}
Since $b(x)\rightarrow \pm k$ as $ x\rightarrow \pm\infty$ (from \eqref{PES}, we have, from \eqref{scalefn} that $p(x) \rightarrow \infty $ as $ x\rightarrow\infty$ and $m([x,\infty)) < \infty$. So, the first term in the RHS of the equation above is finite, and the second term vanishes. As for the third term in the
RHS, we note that
\begin{align}\label{ratio} 
\lim\limits_{n \rightarrow \infty}\frac{1}{p(n)-p(x)}\int_x^n
p(z)m(dz)=
\lim\limits_{n \rightarrow \infty}\frac{2}{\sigma^2}\frac{1}{p(n)-p(x)}\int_x^n\frac{p(z)}{ p'(z)} dz
\end{align} 
Further note that, 
\begin{align}
    \lim\limits_{n \rightarrow \infty}\int_x^n \frac{p(z)}{p'(z)}dz \rightarrow \infty
\end{align}
Using l'Hospitale's rule, we
evaluate the limit in  \eqref{ratio},
\begin{align}
    \lim\limits_{n \rightarrow \infty}\frac{2}{\sigma^2}\frac{1}{p(n)-p(x)}\int_x^n\frac{p(z)}{ p'(z)} dz =\lim\limits_{n \rightarrow \infty} \frac{2}{\sigma^2} \frac{p(n)}{p'(n)^2} = 0,
\end{align}
where the last equality follows from the estimate in \eqref{bestimate} and the definition of $p(y)$.  Hence $E_yT_x < \infty$ when $x < y$.
The case $y < x$ can be
similarly treated by considering $E_y T_x =\lim\limits_{n \rightarrow \infty} E_y(T_{-n} \wedge T_x)$.
\end{proof}

\begin{remark}\label{tau1remark}
Note that positive recurrence implies  $E_a\tau_n < \infty$ for every $n \geq 1$. In
particular, $\tau_n < \infty~ a.s.$ Note that the regenerative property tells us that the cycle times $\tau_n-\tau_{n-1}$ are independent and identically
distributed \cite{kallenberg1997foundations} under $P_a$ with mean $E_a (\tau_n - \tau_{n-1}) = E_a \tau_1 = E_a
T_b + E_bT_a$.
\end{remark}

Before we move on to a discussion of transition paths, we summarize the nomenclature in transition path formalism, largely following \citenum{vanden2010transition}:

Denoting the minima of the PES $U$ as $a$ and $b$ as before, we set $A = (-\infty,a]$ as the \textit{reactant} region and $B = [b,\infty)$ as the \textit{product} region. We define the \textit{last entrance time} into and \textit{first exit time} from the set $\mathbb{R}\backslash A\cup B =(a,b)$, $\sigma_t$ and $\tau_t$ respectively \cite{rajeev1990semi}, as:
\begin{align*}
    \sigma_t &:= \sup_{s\leq t} \left \{s: X_s \notin \mathbb{R} \backslash (A\cup B)  \right \}\\
\tau_t &:= \inf_{s\geq t} \left \{s: X_s \notin \mathbb{R} \backslash (A\cup B)  \right \}
\end{align*}
As in \citenum{vanden2010transition},we denote by $R$ the times at which $X_t \in \mathbb{R}\backslash (A\cup B)$, 
\begin{align}
R := \left \{ t: X_t \in \mathbb{R} \backslash (A\cup B), X_{\sigma_t} \in A \; \text{and} \; X_{\tau_t} \in B \right \} := \bigsqcup_{i=1}^\infty (a_i,b_i)
\end{align}
where the last equality holds because $(X_t)$ is continuous and hence $R$ is an open set. Noting that $X_{a_i} =a , X_{b_i}  =b \; \forall i$, the ensemble of \textit{reactive trajectories} (See figure \ref{Crossings_zoomed}) can be defined as:
\begin{align*}
\mathcal{R} := \bigcup_{i} \left \{ X_t : t \in (a_i,b_i)\right \}
\end{align*}
where each $i$ corresponds to a reactive trajectory. Denoting the number of reactive trajectories until time $t$ as $N_t$, the rate $\kappa$ of the reaction is given as \cite{vanden2010transition}:
\begin{align}\label{reactionrate}
    \kappa := \lim_{t\rightarrow\infty} \frac{N_t}{t} 
\end{align}

It is easy to observe that the process $(Z_t)$ defined as $Z_t := X_t - X_{\sigma_t}$ describes the dynamics of $(X_t)$ in $\mathcal{R}$, as $Z_t \equiv 0$ for $t \notin R$ and $Z_t \in (-(b-a),b-a)$ for $t \in R$ and thus can be considered the \textit{transition path} process. The dynamics of the process $(Z_t)$ merits a separate discussion and in this regards, \citenum{lu2015reactive} conceive the idea of using the Doob-h transform \cite{day1992conditional} to define an auxiliary SDE, whose solutions have the same law as that of the reactive trajectories. We are, however, interested in associating an SDE to the transition path process $(Z_t)$ that focuses on crossing events or \textit{excursions} between the sets $A$ and $B$ as in \citenum{rajeev1989crossings,rajeev1990semi,Rajeev1996}  and making explicit the role of the intrinsic \textit{local time}\cite{kallenberg1997foundations} scale associated with such events. This shall be the subject of the next section. One advantage to this approach, despite the fact that it is not easily extendable to higher dimensions, is that the solutions to such an SDE will be adapted the same filtration $(\mathcal{F}_t)$ as that of the process $(X_t)$. Further discussion on the insights that local times and excursion theory provide in the context of reaction rates is reserved for later sections. 

\section{Singular SDE describing transition path}\label{SSDEsection}

Since the definition of the transition process $(Z_t)$ motivated in the previous section hinges on $(X_t)$, it is expected that the SDE describing $(Z_t)$ will be driven by the continuous martingale $(X_t)$ and not by a brownian motion as in \eqref{SDE}. A detailed description of such SDEs driven by continuous semimartingale and conditions for the existence of unique strong/weak solutions to them has been provided in \citenum{Karandikar2018}. However, the SDEs in \citenum{Karandikar2018} assume a continuous drift term. 

From the discussion in \citenum{rajeev1990semi}, it is clear that the transition path process $(Z_t)$ will have contributions from local time terms, so the SDE describing $(Z_t)$ ought to have a `singular drift' or a `local time drift', the theory of which is explained in \citenum{BassChen10.2307/25053406,BLEI20134337,LeGall10.1007/BFb0099122}. In particular, \citenum{BassChen10.2307/25053406} also discusses the cases in which one can expect a strong solution to such SDEs with singular term. Using the insights in \citenum{rajeev1990semi}, in this section we define a singular SDE that describes the transition path process, explain what a solution to it means, and give the proof of its existence and pathwise uniqueness. 

Given the continuous semimartingale $(X_t)$ that is the strong unique solution to \eqref{SDE}, let $S := \left \{ t: X_t \in (a,b) \right \}$ and consider the following equation:
\begin{align}\label{SSDE}
    dZ_t &= \mathcal{V}(Z_t)dX_t + d\mathcal{L}_t(Z_{.})
\end{align}
with
\begin{subequations}
\begin{align}
\mathcal{V}(x) &= I_{\left \{ 0 < |x| < b-a \right \}}, \\
X_t \notin (a,b) &\Leftrightarrow Z_t \equiv 0 \;\; \text{and} \label{Ztout} \\ 
\supp(d\mathcal{L}_t) \subseteq S^c, &\;i.e. \int_{0}^{\infty}I_S(s)|d\mathcal{L}_s| = 0\label{suppcondn}
\end{align}
\end{subequations}

Since \eqref{SSDE} has a local time drift and the `diffusion coefficient' is the discontinuous function $\mathcal{V}$, this is a singular SDE, driven by a continuous semimartingale. By a \textit{solution} to \eqref{SSDE}, we mean a pair $(Z_t,\mathcal{L}_t)$ such that the following holds:

\begin{enumerate}
    \item $Z_t$ and $\mathcal{L}_t$ are adapted to the filtration $\sigma(X)$, and is right continuous with left limits(rcll), or a \textit{c\`adl\`ag}. 
    \item $\mathcal{L}_t$ has bounded variation with $\mathcal{L}_0 = Z_0 = 0$. Moreover, $\mathcal{L}_t$ is adapted to $\sigma(Z)$, namely the filtration generated by the process $Z$ and satisfies \eqref{suppcondn}. 
    \item $Z_t$ and $\mathcal{L}_t$ satisfies the following piecewise definition: 
    
    Let 
    \begin{align*}
        \tau  = \inf \left \{ u>0: X_u \notin (a,b)\right \}
    \end{align*} 
    
    \begin{enumerate}
        \item For $0 \leq t < \tau$, $Z_t = X_t - X_0$ and $\mathcal{L}_t = 0$. Since $\mathcal{V}(Z_t) =1$ for almost every $(t,\omega) \in [0,\tau) \times \Omega$ w.r.t the measure $dt \; dP$, 
        \begin{align*}
            Z_t = \int_0^t \mathcal{V}(Z_s) dX_s.
        \end{align*}
        \item For $t \geq \tau$, $Z_t = \int_{\tau}^t \mathcal{V}(Z_s) dX_s + \mathcal{L}_t$
    \end{enumerate} 
    
\end{enumerate}

\begin{remark}
Note that $Z_t$ and $\mathcal{L}_t$ are not required to be continuous processes. Since $X_t$ is a continuous process, it follows from \eqref{SSDE} and \eqref{Ztout} that $\Delta Z_t = \Delta \mathcal{L}_t$.
\end{remark}

\newtheorem{theorem}{Theorem}
\begin{theorem} [Existence-Uniqueness theorem]\label{exunthm}
There exists a unique solution to the singular stochastic differential equation \eqref{SSDE}.
\end{theorem}

\begin{proof}
By definition of a solution (condition 3 above) to the singular SDE, the process $(Z_t) = X_t-X_{\sigma_t}$ satisfies \eqref{SSDE} for $t<\tau$. We hence show that the transition path process $Z_t$ is a solution to \eqref{SSDE} for $t\geq\tau$. This follows right away from the Tanaka formula for a semi-martingale, as outlined in \citenum{rajeev1990semi}. We have the following expression for $(Z_t)$ for $t \geq \tau$:

\begin{align}\label{eq4}
\begin{split}
    Z_t = X_t - X_{\sigma_t} &= \int_{\tau}^{t}I_{(a,b]}(X_s)dX_s + \frac{1}{2}\left \{ L(t,a) - L(t,b) \right \} - (b-a)(U(t)-D(t)) \\
    &=\int_{\tau}^{t}I_{(a,b)}(X_s)dX_s + \int_{\tau}^{t}I_{(X_s=b)}dX_s \\&+ \frac{1}{2}\left \{ L(t,a) - L(t,b) \right \} - (b-a)(U(t)-D(t))
\end{split}
\end{align}
where $L(t,a)$ and $L(t,b)$ are local times at $a$ and $b$ respectively of $(X_t)$, $U(t)$ and $D(t)$ are the number of upcrossings and downcrossings of $(a,b)$ by $(X_t)$ during $(0,t)$. Note that since the $dt(\{s:X_s = b\} = 0) \; a.s.$, the second integral in the RHS vanishes. It is easy to observe that $I_{(a,b)}(X_t) = I_{\left \{ 0<|Z_t|<b-a \right \}}$. 

Also, the upcrossings $U(t)$ and downcrossings $D(t)$ of $(a,b)$ by $(X_t)$ are the same as the upcrossings and downcrossings of $(0,b-a)$ and $(-(b-a),0)$ respectively by $(Z_t)$. In fact,  $Z_t = (X_t-a)^+$ during an upcrossing, $Z_t = (X_t-b)^-$ during a downcrossing and jumps to $0$ from $\pm (b-a)$ at the time of crossing (See figure \ref{Crossings} and \ref{Crossings_zoomed}). Let $L_Z(t,x)$ denote the local time of the process $(Z_t)$ at $x$. It is well-known that $L_Z(t,x)$ is right continuous and has left limits at every $x$ \cite{revuz2013continuous}. Further, from \citenum{rajeev1990semi} we can show that the following relation holds: 
\begin{align}\label{LZtLXt}
    L(t,a) = L_Z(t,0)  \;\;\;\;\;\;\; L(t,b) = L_Z(t,0-)  
\end{align}

\begin{figure}
    \hspace{-2cm}
    \includegraphics[scale=0.83,left]{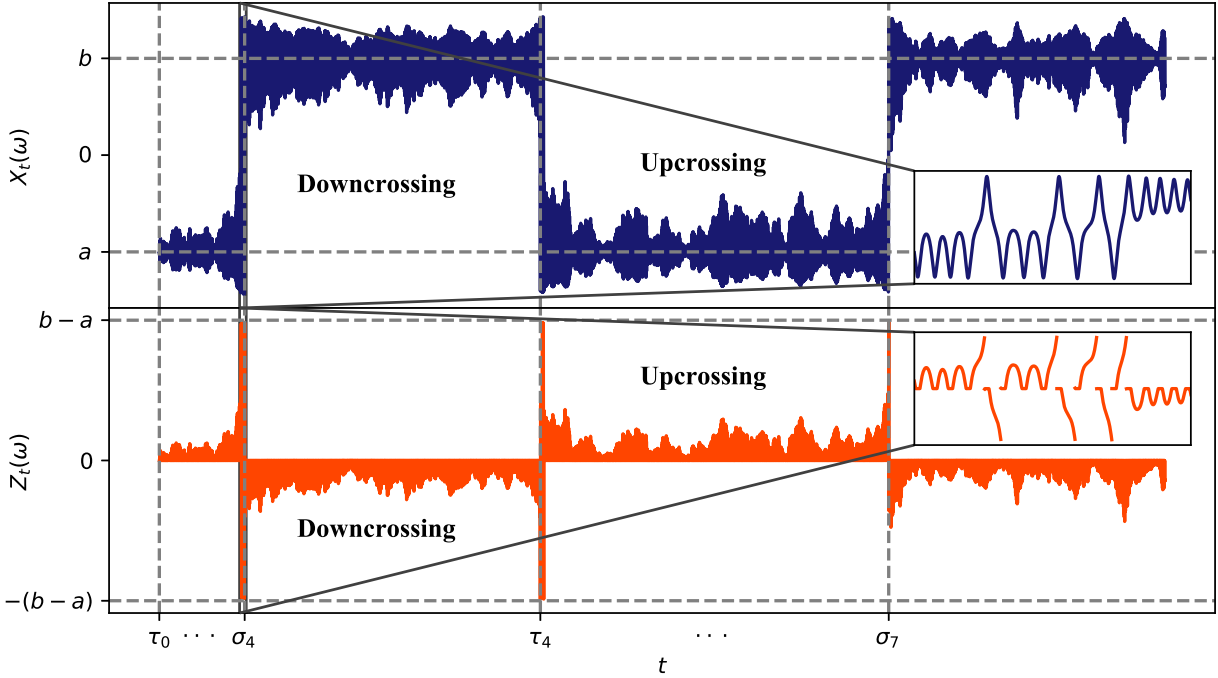}
    \caption{Illustration of the transition path process $(Z_t)$: One specific realization of the diffusion $X_t(\omega)$ and the corresponding transition path process $Z_t(\omega)$ are plotted above. Here, each $(\tau_i,\sigma_{i+1})$ is an upcrossing interval and $(\sigma_i,\tau_{i+1})$ is a downcrossing interval. As can be seen from the inset image, $Z_t(\omega)$ is discontinuous at the end points of these intervals (See figure \ref{Crossings_zoomed} for further clarity). }
    \label{Crossings}
\end{figure}

The first equation in \eqref{LZtLXt} follows from eqn. 7 in \citenum{rajeev1990semi}. The second equation in \eqref{LZtLXt} follows from letting $x$ increase to $0$ in eqn. 6  (We use the notations as in \citenum{rajeev1990semi} and observe that $L(t,b)$ has no support in the upcrossing intervals (i.e. in $\{ s: \Theta^u(s) =1 \}$). This shows that
\begin{align}
   \mathcal{L}_t(Z_.) := \frac{1}{2}\left \{ L(t,a) - L(t,b) \right \} - (b-a)(U(t)-D(t)) 
\end{align}
is a well-defined functional of the process $(Z_t)$. It is straightforward to show that $\mathcal{L}_t$ has bounded variation and the support property in \eqref{suppcondn} is satisfied. As $X_t = (X_t - X_{\sigma_t}) + X_{\sigma_t} =: Z_t + Y_t $, we have that $(Z_t,\mathcal{L}_t)$ is a solution to \eqref{SSDE}.
 
We prove the uniqueness of the solution in two parts. Suppose there is another solution $(Z_t^{'},\mathcal{L}_t^{'})$ to equation \eqref{SSDE}. It suffices to show that $Z_t^{'} = Z_t \;\; \forall t \; a.s$. It then follows that $\mathcal{L}^{'}_t (Z') = \mathcal{L}_t (Z)$. For $t \notin S$, $X_t \notin (a,b)$ and hence $Y_t^{'} = X_t = Y_t$ as $Z_t^{'}$ satisfies \eqref{Ztout}. Now, for the case of $t \in S$, we proceed via a pathwise argument.
\begin{align}\label{uni2}
    Z^{'}_t = \int_0^t I_{\{0<|Z^{'}_t| < b-a\}} dX_s + \mathcal{L}^{'}_t(Z^{'}_t)
\end{align}

\begin{figure}
    \hspace{-1.75cm}
    \includegraphics[scale=0.83,left]{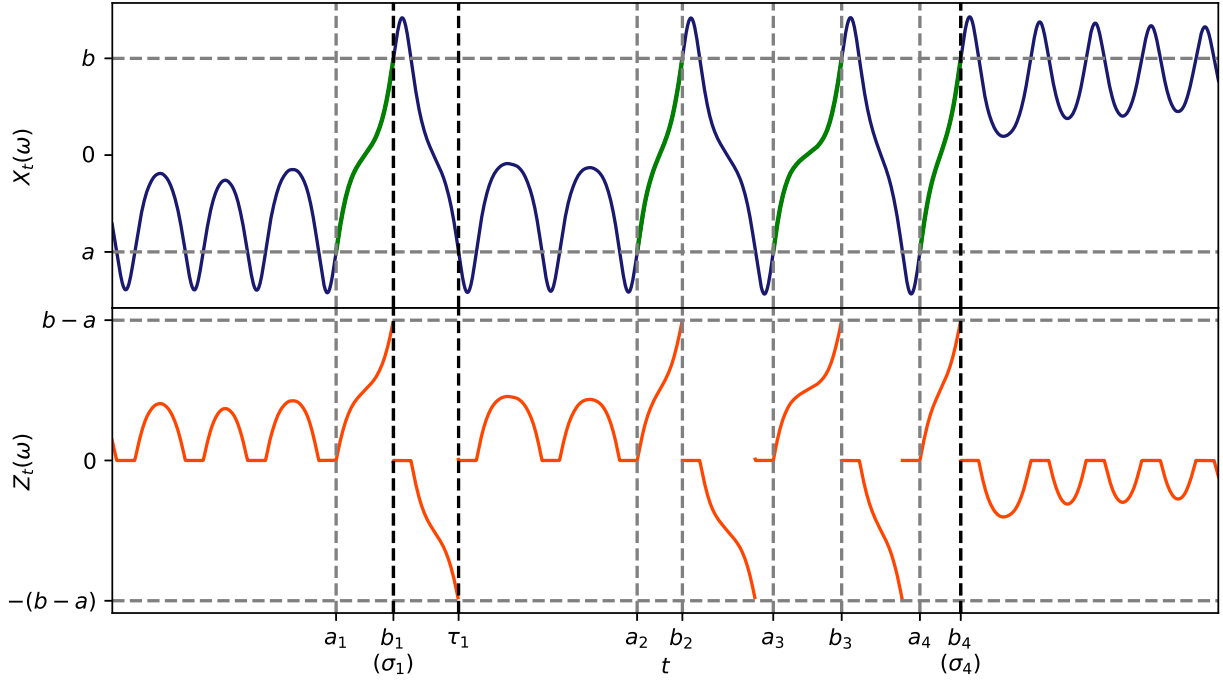}
    \caption{The inset image in figure \ref{Crossings} is zoomed further. The portions of $X_t(\omega)$ highlighted in green are the reactive trajectories. Each reactive trajectory occurs in the time interval $(a_i,b_i)$. The points $\sigma_1, \tau_1$ and $\sigma_4$ are also marked in black for reference. The discontinuities in $Z_t(\omega)$ at each $\sigma_i$ and $\tau_i$ are clearly evident.}
    \label{Crossings_zoomed}
\end{figure}

Denoting the set $S^{'} := \{s : 0 < |Z^{'}_s| < b-a\}$,  we write $S :=\bigcup_i (s_i,t_i) = S^{'} $. For $h>0$, define 
\begin{align}
    D_h := \inf \{ s>h : |Z^{'}_s| \notin (0,b-a) \}
\end{align}
From the positive recurrence of $(X_t)$, we can show that $D_h < \infty$ a.s.
Hence we have for a fixed $h>0$, from \eqref{uni2} and the fact that $\supp(d\mathcal{L}^{'}_t) \subseteq S^c $, 
\begin{align}\label{hDh}
    Z^{'}_{h} - Z^{'}_{D_h} = X_{h} - X_{D_h} \;\; \text{a.s.}
\end{align}

For $\omega$ s.t. $h<D_h(\omega)$ and $h \in (s_i(\omega),t_i(\omega))$, take $h_n \in (s_i(\omega),t_i(\omega))$, then we have from \eqref{hDh} applied to $h$ and $h_n$ separately,
\begin{align}
    Z^{'}_h - Z^{'}_{h_n} = X_h - X_{h_n}
\end{align}
since $D_h = D_{h_n}$. Taking $h_n < h$, $h_n\downarrow \sigma_h$ and noting that $Z^{'}_{\sigma_h} = 0$, we get 
\begin{align}\label{Zh'}
    Z^{'}_h = X_h - X_{\sigma_h} \;\; \forall h \in (s_i(\omega),t_i(\omega))
\end{align}

\eqref{Zh'} and the fact that $\Omega = \bigcup_{h \in Q^+} \{\omega : h < D_h(\omega) \} $ completes the proof. 

\end{proof}

Thus we have characterized the transition path process as the unique solution of an associated singular stochastic differential equation. The fact that the transition path process $(Z_t)$ is described by a  singular SDE involving local time drift suggests that there are local timescales associated with $(Z_t)$ that one needs to take cognizance of. In the next section, we will explore this connection in detail and provide an excursion theoretic conceptualization of reactive trajectories. 

\section{Excursions of transition paths and reaction rates}\label{excurthsect}

The transition path process $(Z_t)$ characterizes the shuttling action of the process $(X_t)$ into and outside the set $\mathbb{R}\backslash(A\cup B)$, which is $(a,b)$ in our discussion. \textit{Ito excursion theory} \cite{ito1972poisson} describes the Poisson point process of excursions from a point associated with the continuous recurrent diffusion such as $(X_t)$, using the local time at that point. In the context of this paper, we need a theory of excursions of the semimartingale $(X_t)$ into the set $(a,b)$ starting from the boundary $\left \{ a,b \right \}$.  This need a suitable modification of Ito's theory to include excursions from more than one point. For a generalization of Ito's excursion theory, we refer the reader to Maissonneuve's discussion on \textit{Exit systems} \cite{maisonneuve1975exit}.

For the specific case of the set $B$ being $(a,b)$ and $(X_t)$ being a Brownian motion, \citenum{bhaskaran2022asymptotic} had described the associated \text{excursions} of $(X_t)$ in terms of Ito point excursion processes and local times at the boundary points of the interval $(a,b)$. It was shown that the Maissonneuve excursion of $(X_t)$ into the set $(a,b)$ is `decomposable' into Ito excursions of $(X_t)$ from the points $a$ and $b$. This section shall focus on extending the arguments in \cite{bhaskaran2022asymptotic} to describe the excursions of the continuous diffusion $(X_t)$ about $a$ and $b$, which can be equivalently formulated in terms of the excursions of the transition path $(Z_t)$
about the point $0$. We use the excursion theory presented in \citenum{kallenberg1997foundations} for regenerative processes. 

We begin this section with the definitions of \textit{excursion space} and \textit{excursion processes} and then proceed to give an expression for the rate constant in terms of excursions of the processes $(X_t)$ or equivalently that of $(Z_t)$.

The space of excursions of $(X_t)$ into the set $(a,b)$ is defined as follows:
\begin{equation*}
\begin{split}
    \mathcal{E}(a,b) := \{ u | \; &u: [0,R(u)) \rightarrow [a,b], \; u \text{ is continuous on} \; [0,R(u)), \\&a<u(t)<b \; \forall t \in (0,R(u) ), \; u(0) \in \left \{ a,b \right \} \; \text{and} \; u(R(u)-) \in \left \{ a,b \} \right \} .
\end{split}
\end{equation*}
$\mathcal{E}(a,b)$ can be identified as a measurable subset of the space $C([0,\infty))$ with the induced topology by extending each function $u \in \mathcal{E}(a,b)$ as a constant beyond $R(u)$. Note that $\mathcal{E}(a,b) = \mathcal{E}_a \bigcup \mathcal{E}_b$, where, for $x\in \{a,b\}$,
\begin{align*}
    \mathcal{E}_x := \{ u \in \mathcal{E}(a,b)|  \; u(0)=x  \}.
\end{align*}
The same space can be expressed in terms of excursions of the transition path process $Z_t$:
\begin{equation*}
\begin{split}
    \mathcal{E}_Z (\mathcal{S}) := \{u| \; &u: [0,R(u))\rightarrow \mathcal{S} \cup \{0 \},\; \text{u is continuous in} \; [0,R(u)),\\ &u(0)=0, \; u(R(u)-) \in \left \{ 0,b-a, -(b-a) \right \}, u(t) \in \mathcal{S} \;\forall t \in (0,R(u)) \; \} ,
\end{split}
\end{equation*}
with $\mathcal{S} := (-(b-a),0) \cup (0,b-a)$. As in the case of $\mathcal{E}(a,b)$, we can write $\mathcal{E}_Z(\mathcal{S}) = \mathcal{E}_0^+ \cup \mathcal{E}_0^-$, where $\mathcal{E}_0^+$ and $\mathcal{E}_0^-$ are the positive and negative excursions from $0$. 

For the continuous, non decreasing, $\mathcal{F}_t$-adapted process $L_t:= L(t,a) + L(t,b)$, the right continuous inverse $\tau_t$ can be defined as follows:
\begin{align*}
   \tau_t:= \inf\{s>0: L_s > t\} 
\end{align*}
The right continuous inverse $\tau_t^a$ and $\tau_t^b$ of $L(t,a)$ and $L(t,b)$ respectively and the right continuous inverse $\tilde{\tau}_t$ of the local time process at $0$ for the process $(Z_t)$, namely $L_Z(t,0)$, can be defined similarly. 

\begin{prop}
$L_Z(t,0) + L_Z(t,0-) = L(t,a) + L(t,b) = L_t$. 
\end{prop}
\begin{proof}
The proof follows from \citenum{rajeev1990semi} (See the proof of theorem \ref{exunthm} above).
\end{proof}

We now define the excursion process; setting $D := \{s>0: \tau_s \neq \tau_{s-} \}$, for $t \in D$, we define:
\begin{align*}
    \widetilde{e}_t(\omega)(s):=X_{\tau_{t-} + \; s\wedge (\tau_t-\tau_{t-}) }(\omega).
\end{align*}

For $t \notin D$, $\widetilde{e}_t(\omega) \equiv \delta$. The above definition of excursion includes all the excursions starting from $a$ and $b$. However, we only need those excursions starting from $a$ or $b$ into the set $(a,b)$. Hence we discard the excursions below $a$ and above $b$ by redefining the excursion process as:
\begin{align*}
    e_t(\omega) := I_{\mathcal{E}(a,b)}(\widetilde{e}_t(\omega))\;\widetilde{e}_t(\omega) + \delta I_{(\mathcal{E}(a,b))^c}\;(\widetilde{e}_t(\omega)).
\end{align*}
where $I_A$ denotes the indicator function of the set $A$. In terms of $(Z_t)$, we can redefine the excursion process as:
\begin{align*}
    e_t^Z(w):=I_{\mathcal{E}_Z(\mathcal{S})}(\widetilde{e}_t^Z(\omega))\;\widetilde{e}_t^Z(\omega) + \delta I_{(\mathcal{E}_Z(\mathcal{S}))^c}\;(\widetilde{e}_t^Z(\omega))
\end{align*}
where $\widetilde{e}_t^Z(\omega)$ is defined similar to $\widetilde{e}_t(\omega)$, by replacing $(X_t)$ by $(Z_t)$. 

Let $\mathcal{E}'_a$ be the set of excursions from $a$, that is, 
\begin{align*}
    \mathcal{E}'_a = \{ u: [0,R(u)) \rightarrow \mathbb{R}, u \; \text{is continuous on} \; [0,R(u))\; \\ \text{and} \; u(t) \neq a \; \text{for} \; t \in (0,R(u)), \; u(0) = u(R(u)-) = a \}
\end{align*}
and 
\begin{align*}
(\mathcal{E}'_a)^+ = \{ u \in \mathcal{E}'_a | u(t) > a \;\; \text{for} \; t \in (0,R(u))\}
\end{align*}
is the set of positive excursions starting from $a$. Defining $(\mathcal{E}'_a)^-$ similarly, we have $\mathcal{E}_a' = (\mathcal{E}'_a)^+ \bigcup (\mathcal{E}'_a)^-$.
Similar definitions hold for the set of excursions starting from $b.$ 

Define the map $t_b(u)$ for $u \in (\mathcal{E}_a')^+$ as follows:
\begin{align*}
    t_b(u) := \inf \; \{ t | \; u(t) = b \}.
\end{align*}
Now define 
\begin{align*}
    \eta_a := \; (\mathcal{E}_a')^+ &\rightarrow \mathcal{E}_a \\
     u &\mapsto \eta_a(u) := u(t_b \wedge \cdot)
\end{align*}
We can similarly define $t_a$ and $\eta_b$. 

Let $n'_a$, $n'_b$ be the characteristic measures on $\mathcal{E}_a'$ and $\mathcal{E}_b'$ respectively of the Poisson point process $N_a$ and $N_b$ associated with the excursions of $(X_t)$ starting from $a$ and $b$ respectively. Define $n(x,A)$ on $\mathcal{E}(a,b) = \mathcal{E}_a \cup \mathcal{E}_b$ as follows: 
\begin{align}
    n(x,A) := n'_x(\eta_x^{-1}(A \cap \mathcal{E}_x)) , \; \text{where} \; x \in \{a,b\} \; \text{and} \; A \subseteq \mathcal{E}(a,b)
\end{align}

Using these definitions of excursion process, we can now characterize reactive trajectories. First we set:
\begin{align*}
\Lambda &:=\{u \in \mathcal{E}(a,b) : u(0)=a, \underset{0<t<R(u)}\sup u(t) =b \} \subseteq \mathcal{E}_a\\
\Lambda_Z &:= \{u \in \mathcal{E}_Z(\mathcal{S}): \underset{0<t<R(u)}\sup u(t) =b-a \} \subseteq \mathcal{E}_0^+
\end{align*}
Then, we define:
\begin{align*}
    \Lambda' := \{ u \in (\mathcal{E}_a')^+| u(0) = a, \underset{0<t<R(u)}\sup u(t) \geq b \} = \eta_a^{-1}(\Lambda )  
\end{align*}
With this, we define the quantity $N(t,A)$ for $\; A \subseteq \mathcal{E}(a,b)$: 
\begin{align*}
    N(t,A) := \# \{ s\leq t: \tau_{s-} \neq \tau_s, e_s(\omega) \in A  \} = N_a(t, \eta_a^{-1}(A)) + N_b(t,\eta_b^{-1}(A))
\end{align*}
$N_Z(t,A)$ can be defined similarly. For computing the rate of the reaction we need to calculate the number of reactive trajectories, $N_t$ until time $t$.

\begin{prop}
 The number $N_t$ of reactive trajectories until time $t$, is given as follows:
  \begin{align*}
       N_t = N(L_t,\Lambda) &= N_a(L(t,a),\Lambda'), \\
      where \;\; L_t  = L(t,a) & + L(t,b), \; \text{ as  defined  above.}
  \end{align*}

\end{prop}

\begin{proof}
    By construction, the support of the local time process $dL_t$,
    \begin{center}
        $\supp(dL_t) \subseteq \left \{ t: X_t = a \;\; or \;\; X_t = b \right \}$.
    \end{center} 
    Note that $L_t$ is constant during the excursion interval $(\tau_{t-},\tau_t)$. The number of reactive trajectories until the time $t$ is, by definition, equal to the number of upcrossings of $(X_t)$ completed before $t$. So, each upcrossing of $(a,b)$ by $s \mapsto X_s(\omega)$ completed before $t$ corresponds to an excursion $e_s(\omega)$ of $(X_t)$ into the set  $ \Lambda$, with $\tau_s \leq t$. Since $\tau_t$ is the right inverse of the local time process $L_t$, we have that:
    \begin{align*}
        N_t = N({L_t},{\Lambda})
    \end{align*}
    Clearly the number of excursions from $a$ of height $b-a$ counted via $L_t$ and $L(t,a)$ must be the same. So we have:
    \begin{align*}
        N_t = N_a({L(t,a)},\Lambda')
    \end{align*}
    
\end{proof}

\section{Computing reaction rate using excursion theory}\label{reactionratesect}

The rate $\kappa$ of the reaction mentioned in \eqref{reactionrate} can be expressed as:
\begin{align*}
    \kappa = \underset{t\rightarrow\infty}\lim \frac{N_t}{t} = \underset{t\rightarrow\infty}\lim \frac{N(L_t,\Lambda)}{t} =
    \underset{t\rightarrow\infty}\lim\frac{N_a(L(t,a),\Lambda')}{t}
\end{align*}
This can be rewritten as:
\begin{equation}\label{ExcRate}
    \kappa = \underset{t\rightarrow\infty}\lim\frac{N(L_t,\Lambda)}{L_t} \; \underset{t\rightarrow\infty}\lim\frac{L_t}{t} = 
    \underset{t\rightarrow\infty}\lim\frac{N_a(L(t,a),\Lambda')}{L(t,a)}\;
    \underset{t\rightarrow\infty}\lim\frac{L(t,a)}{t} 
\end{equation}
provided both the limit exists (Note that $N(t,A) \leq \infty$ for an arbitrary set $A \in \mathcal{E}(a,b)$. However, $N_a(t,A) < \infty$ a.e. iff $n'_a(A) < \infty$). We show that is the indeed the case.

\begin{prop}\label{2limits}
\begin{align*}
    \underset{t\rightarrow\infty}\lim\frac{L(t,a)}{t} = \frac{E_aL(\tau_1,a)}{E_a\tau_1} = \frac{ -2 E_a \int_0^{\tau_1}I_{\{X_s>a\}}b(X_s)ds}{E_a \tau_1} < \infty
\end{align*}
\end{prop}
\begin{proof}
The proof is given in 2 steps: \\
\textbf{Step 1:} To show 
\begin{align}
\underset{n\rightarrow\infty}\lim\frac{L(\tau_n,a)}{n} = E_a L(\tau_1,a) < \infty 
\end{align}
where $\tau_n$ is as defined in section \ref{Prelims}.

\hspace{-0.5cm}\textbf{Step 2:} To show 
\begin{align}
    \underset{t\rightarrow\infty}\lim\frac{L(t,a)}{t} = \frac{E_a L(\tau_1,a)}{E_a\tau_1}
\end{align}

\textit{Proof of step 1:}
That the limit in the LHS exists follows from a standard `Law of large numbers' argument; denoting $\Xi_k := L(\tau_k,a) - L(\tau_{k-1},a) $, we have, by the regenerative property:
\begin{align*}
    L(\tau_n,a) = \sum_{k=1}^n \Xi_k
\end{align*}where 
$\Xi_k$ are independent and identically distributed with
\begin{subequations}\label{ExpLt}
\begin{align}
  E_a \Xi_1 &= E_a L(\tau_1,a)  \\
  &= 2E_a \left ( (X_{\tau_1}-a)^+ - (X_0 -a)^+ - \int_0^{\tau_1} I_{\{X_s>a\}} dX_s \right ) \\
  &=   -2 E_a \int_0^{\tau_1}I_{\{X_s>a\}}b(X_s)ds = -\frac{2}{\sigma^2} \int_a^{\infty}b(y)\;E_aL(\tau_1,y)dy
\end{align}
\end{subequations}

In the expression above, the second equality follows from the Tanaka formula \cite{kallenberg1997foundations} and the last equality follows from the `Occupation density formula'(See Chapter 3, theorem 7.1(iii)  in \citenum{karatzas2012brownian}) for the semimartingale $(X_t)$ and a bounded Borel function. Note that the expectations are finite because $E_a \tau_1 < \infty$ and $b$ is a bounded function. This completes the proof of step 1. 


\textit{Proof of step 2:}
Using a standard argument, we can write:
\begin{align*}
    \frac{L(\tau_n,a)}{\tau_{n+1}} &\leq \frac{L(t,a)}{t} \leq \frac{L(\tau_{n+1},a)}{\tau_{n}}
\end{align*}

for $\tau_{n} \leq t \leq \tau_{n+1}$.

Now the proof follows using step 1 and the fact observed earlier in the proof of proposition \ref{posrecur} and remark \ref{tau1remark} that \begin{align*}
    \lim_{n\rightarrow\infty} \frac{\tau_n}{n} \rightarrow E_a \tau_1
\end{align*}
\end{proof}

\begin{remark}
Considering \eqref{ExpLt} for an arbitrary point $z \in \mathbb{R}$, we get:
\begin{align*}
    E_a L(\tau_1,z) = -\frac{2}{\sigma^2} \int_z^{\infty} b(y) E_a L(\tau_1,y) dy
\end{align*}
\end{remark}
Setting $f(z):= E_a L(\tau_1, z)$, this yields an ordinary differential equation, whose general solution is given as $f(z) \sim \exp(\sfrac{2}{\sigma^2}{\int_0^z  b(x) dx})$, thereby yielding:
\begin{align}\label{ExpLtU}
    E_a L(\tau_1, z) = C e^{-\beta (U(z) - U(0))}
\end{align}
where $\beta= \sfrac{1}{k_B T}$. Thus, the expected value of local time is highest near the minima of the PES $U(x)$, which are the points $a$ and $b$ and explicitly:
\begin{align}\label{ExpLta}
    E_a L(\tau_1, a) \sim e^{-\beta (U(a) - U(0))} = e^{\beta V_b}
\end{align}
where $V_b$ is the height of the potential energy barrier in $U(x)$, as $0$ corresponds to a local maxima of $U(x)$. In the above expression, $a \sim b \iff a=C b$ for some constant $C$.

\begin{prop}
The excursion measure of $\Lambda'$, namely $n'_a(\Lambda')$, is finite and in particular, $(N_a(t,\Lambda'))_{t\geq0}$ is a Poisson process with
\begin{align*}
    EN_a(t,\Lambda') = t \; n'_a(\Lambda')
\end{align*}
\end{prop} 

\begin{proof}
For a Poisson random measure $\xi$, we have $\xi(A) < \infty \; a.s.$ if and only if $\mu(A):=E\xi(A)$ is finite. This follows from the `exponential formula' (Refer to lemma 12.2 part (i) in \citenum{kallenberg1997foundations}): 
\begin{align*}
E e^{-\xi(A)} = e^{-\int (1 - e^{-I_A}) d\mu}
\end{align*} 

Let $\Lambda^a_h$ be the set of excursions of length $h$ from $a$, with $\Lambda^a_h = \Lambda_h^{a,+} \bigcup \Lambda_h^{a,-}$ where $\Lambda_h^{a,+}$ and $\Lambda_h^{a,-}$ are the positive and negative excursions from $a$.  Note that $n'_a(\Lambda^a_h) < \infty$ for $h >0$ (Refer \citenum{kallenberg1997foundations} for proof). We define Poisson random measures $\Xi_h(\cdot) := N_a(t,\cdot \cap \Lambda_h^a)$ for $h>0$ \cite{kallenberg1997foundations} and note that $\Xi_h(\Lambda') < \infty \; a.s.$, as $n'_a(\Lambda_h^a) < \infty$ for each $h>0$. Defining a sequence  $\Xi_{h_n}$ of Poisson random measures with $h_n \downarrow 0$, we have:
\begin{align*}
    E\:\Xi_{h_n}(\Lambda') = t\; n'_{a}(\Lambda'\cap \Lambda^a_{h_n}) 
\end{align*}

The sequence $\Xi_{h_n}(\Lambda')$ converges $a.s.$ by the monotone convergence theorem, to a random variable $\Xi(\Lambda')$ for which,
\begin{align}\label{expformula}
Ee^{-\Xi(\Lambda')} &= \lim_{n\rightarrow\infty} E e^{-\Xi_{h_n}(\Lambda')} = \lim_{n\rightarrow\infty} e^{-\int (1 - e^{-I_{\Lambda'}}) d\mu_{h_n}} \\ \nonumber
&= \lim_{n\rightarrow\infty}\exp{(-\mu_{h_n}(\Lambda')(1-e^{-1}))} = \exp{(-\mu(\Lambda')(1-e^{-1}))}
\end{align}

Now, $\Xi(\Lambda') = N_a(t,\Lambda')$ by definition. So, by the continuity of the trajectories of $(X_t)$,
\begin{align*}
    N_a(t,\Lambda') &= \text{Number of upcrossings of $(a,b)$ by the process $(X_s)$} \\ &\;\;\;\;\;\text{during the interval $(0,\tau_t)$} \\
    \Rightarrow N_a(t,\Lambda') &  < \infty \; \textit{a.s.}
\end{align*}

Since the LHS in \eqref{expformula} is positive, it follows that $\mu(\Lambda') <\infty$. Thus, the sequence $\Xi_{h_n}(\Lambda')$ converges in distribution to $\Xi(\Lambda')$. That $\Xi(\Lambda')$ is a Poisson random variable follows from \eqref{expformula}. 

Also, we have: 
\begin{align*}
    \mu(\Lambda') \equiv E\: \Xi(\Lambda') = \lim_{n\rightarrow\infty} E \: \Xi_{h_n}(\Lambda') = t \; \lim_{n\rightarrow\infty} n'_a(\Lambda'\cap \Lambda^{h_n}_a) = t \; n'_a(\Lambda') < \infty 
\end{align*}
From the independent increment property of $(N_a(t,\Lambda'\cap\Lambda^a_h))_{t\geq0}$, the corresponding property of $(N_a(t,\Lambda'))_{t\geq0}$ follows by letting $h\rightarrow0$.
\end{proof}

\begin{cor}
The first limit in \eqref{ExcRate} exists with
\begin{equation}
    n'_a(\Lambda') =  \underset{t\rightarrow\infty}\lim \frac{N_a{(t,\Lambda'})}{t} = \underset{t\rightarrow\infty}\lim \frac{N_a(L(t,a),\Lambda')}{L(t,a)}, \;\; a.s.
\end{equation}
\end{cor}

\begin{proof}
It follows from the independent increment property of the process $(N_a(t,\Lambda'))_{t\geq0}$ and the law of large numbers that 
\begin{align*}
    \lim_{k\rightarrow\infty} \frac{N_a(k,\Lambda')}{k} = n'_a(\Lambda'), \;\; a.s.
\end{align*}
The first equality in the statement follows by the usual interpolation argument as in the proof of proposition \ref{2limits}. The second equality in the statement follows from the first equality and the fact that $L(t,a) \rightarrow \infty$ monotonically as $t\rightarrow\infty$, $a.s$. The null set in the statement is obtained by the intersection of the null sets in the above two steps. 
\end{proof}

\begin{remark}
The factorization of the reaction rate expression \eqref{kappa} for $\sfrac{N_t}{t}$ as $t \rightarrow \infty$ in the positive recurrent case contrasts sharply with the null recurrent case of Brownian motion (see  \citenum{rajeev2013brownian}).

\end{remark}

Thus, \eqref{ExcRate} factorizes the reaction rate into a product of a local time term  and an excursion measure term $n'_a(\Lambda')$. This suggests that the local time at $a$ represents the fluctuations of the process $(X_t)$ around $a$, and the excursion measure captures the `transition probability' of going from $a$ to $b$, which is consistent with the expression for reaction rate. In the next section, we follow a alternate procedure using renewal theory, familiar in the theory of discrete Markov processes, to obtain the rate of the reaction. Using the expression thus obtained, we obtain an explicit expression for the excursion measure $n'_a(\Lambda')$. 


\section{Transition path excursions as renewal events}\label{renewalthsect}

Regenerative processes, such as $(X_t)$, can be decomposed into independent, identically distributed blocks with varying path length, and an analysis that exploits this structure of a regenerative process is provided by \textit{renewal theory} \cite{resnick1992adventures}. Such blocks satisfy the \textit{renewal limit theorem}, which allows us to calculate the reaction rate in \eqref{reactionrate} which shall be the concern of this section. We first begin with a brief introduction to the theory of discrete renewal processes following \citenum{resnick1992adventures}. 

Consider a sequence $\{Y_n\}_{n\geq0}$ of independent, positive, real valued random variables such that $\{Y_n\}_{n\geq1}$ is identically distributed with a distribution $F$, with $F(0-)=0$ and $F(0)<1$. Define the \textit{renewal sequence}, for $n\geq0$,
\begin{align*}
    S_n = \sum_{i=0}^{n}Y_i
\end{align*}
The \textit{counting function}, which counts the number of renewals until a given time $t$, is:
\begin{align*}
    C(t) := \sum_{n=0}^{\infty} I_{[0,t]}(S_n)
\end{align*}
The expectation of the counting function is the \textit{renewal function}, 
\begin{align*}
    \mu(t) := EC(t)
\end{align*}
If the distribution of $Y_n = S_n-S_{n-1}, n\geq1$ has finite mean $\mu$, then we have the following \textit{Renewal theorem}:

\begin{theorem}{Renewal theorem}\\
If the distribution $F$ of $Y_i,i\geq1$  has finite mean, i.e. $EY_1 = \mu <\infty$, then we have the following:
\begin{align*}
    \underset{t\rightarrow \infty}\lim\frac{C(t)}{t} = \underset{t\rightarrow \infty}\lim\frac{\mu(t)}{t}  = \frac{1}{\mu }
\end{align*}
\end{theorem}

An important consequence of the regenerative property of the process $(X_t)$ is the following `cycle decomposition', which is sometimes taken as an alternative definition of regenerative processes \cite{resnick1992adventures}.  

\begin{prop}\label{CycleDecomp}
The process $\{X_t\}_{ t\geq 0}$ has the following cycle decomposition: There exists a renewal sequence $\{S_n\}_{n\geq0}$, such that $S_n \rightarrow \infty$, and for every $k, 0<t_1<...<t_k$, $B \in \mathcal{B}(\mathbb{R}_+^k)$, $A \in \mathcal{B}(\mathbb{R}_+^\infty)$, $\forall n\geq0$, we have:
\begin{align*}
    P_a[(X(S_n+t_i), i=1,2...,k) \in B, \{S_{n+i}-S_n,i\geq1\} \in A | \;S_0,... S_n] = \\
    P_a[(X(t_i), i=1,2...,k) \in B, \{S_{i}-S_0,i\geq1\} \in A]
\end{align*}
\end{prop}

\begin{proof}
We take $S_n:= \tau_n$. In the strong Markov property, we take $F$ as follows:
\begin{align}
    F :=  I_{\{(X_{t_1},X_{t_2}...X_{t_k}) \in B, \; (S_1, S_2...) \in A \}}
\end{align}

Now, we note that 
\begin{align}
    F \circ \theta_{\tau_n} = I_{\{(X_{S_n +t_1},X_{S_n + t_2}...X_{S_n + t_k}) \in B, \; (S_{n+1}-S_n, S_{n+2}-S_n...) \in A \}}
\end{align}

Integrating \eqref{strongMarkov} over a set $G \in \sigma{\{S_0,S_1...S_n\}} \subset \mathcal{F}_{\tau_n} $, we get the equation in the proposition.  

\end{proof}

Thus, $(X_t)$ allows a cycle decomposition with $\left \{ {\tau_n }\right \}_{n\geq0} $ as the renewal sequence. From section \ref{Prelims}, it is clear that $\tau_n $ correspond to the time of the $n^{\text{th}}$ \textit{downcrossing} of $(X_t)$ from $b$ to $a$, so 
\begin{align*}
    C^{(d)}(t) := \sum_{n=1}^{\infty}I_{[0,t]}(\tau_n)
\end{align*}
becomes the counting function $C(t)$ corresponding to this renewal sequence, for which the renewal limit theorem holds. For computing the reaction rate $\kappa$ in \eqref{reactionrate}, we need the number $N_t$ of reactive trajectories until time $t$, which will correspond to the number $C^{(u)}(t)$ of \textit{upcrossings} of $(X_t)$ from $a$ to $b$ until time $t$. But, since there cannot be a downcrossing event without a corresponding upcrossing, we have:
\begin{align*}
    |{C^{(u)}(t) -C^{(d)}(t)}| \leq  1
\end{align*} 
Denoting $\mu = E (\tau_1 -\tau_0)$, we have:
\begin{align*}
    \underset{t\rightarrow \infty}\lim\frac{N_t}{t}& = \underset{t\rightarrow \infty}\lim\frac{C^{(u)}(t)}{t} \\ 
    &=\underset{t\rightarrow \infty}\lim\frac{C(t)}{t} \\
    &= \frac{1}{\mu}
\end{align*}
Thus, from the renewal theorem, we get the required expression for the rate of the reaction:
\begin{align}\label{kappaRenewalthm}
    \kappa = \lim_{t \rightarrow \infty}\frac{N_t}{t}  =\frac{1}{\mu} = \frac{1}{E_a \tau_1}
\end{align}
where the last equality follows from remark \ref{tau1remark}. The expression for reaction rate obtained from excursion theory \eqref{ExcRate} and renewal theory should clearly be the same, so we have:
\begin{align}\label{kappa}
    \kappa = \frac{1}{\mu} = \frac{1}{E_a \tau_1} = n'_a(\Lambda') \;\; \frac{E_a L(\tau_1,a)}{E_a \tau_1}
\end{align}
This expression implies that the product of excursion measure $n'_a(\Lambda')$ and the expected local time at $a$ is unity. From \eqref{ExpLtU}, we hence obtain:
\begin{align*}
    n'_a(\Lambda') \sim e^{-\beta (U(0)-U(a))} = e^{-\beta V_b}
\end{align*}
Thus, renewal theory allows the computation of an explicit expression for the excursion measure of the set $\Lambda'$ corresponding to transition path excursions from $a$ to $b$.

The \textit{forward} and \textit{backward} rate constant, $k_{A,B}$ and $k_{B,A}$ respectively, can be computed from \eqref{kappaRenewalthm} and remark \ref{tau1remark} using eqn. (34) and (35) in \citenum{vanden2010transition}. In the context of our paper, it follows as in the proof of proposition \ref{2limits} that $\rho_A$, the proportion of time spent by $(X_t)$ in upcrossings viz. $\sfrac{\mathcal{T}_U(t)}{t}$ as $t\rightarrow\infty$ is given as:
\begin{equation}
    \rho_A := \lim_{t\rightarrow\infty} \frac{\mathcal{T}_U(t)}{t} =  \frac{E_aT_b}{E_a\tau_1}
\end{equation}
where
\begin{align}
    \mathcal{T}_U(t) := \int_0^t \sum_{i=1}^\infty I_{(\tau_{i-1},\sigma_i)}(s) ds 
\end{align}
is the time spent in upcrossings during $(0,t)$. This yields the following expression for the forward rate constant:
\begin{align}\label{forwardrate}
    k_{A,B} = \frac{\kappa}{\rho_A} = \frac{1}{E_aT_b} = n'_a(\Lambda') \;\; \frac{E_a L(\tau_1,a)}{E_aT_b}
\end{align}

Moreover, in the context of Kramers reaction rate theory, $E_a T_b$ has been interpreted as the \textit{mean first passage time} \cite{hanggi1990reaction}, that is, the average time taken by the diffusion $(X_t)$ to leave the potential well at $a$ in $U(x)$ and reach $b$. In the overdamped regime that we consider, an explicit expression has been obtained for the same \cite{hanggi1990reaction}:
\begin{align}
    E_a T_b = \frac{1}{D}\frac{2\pi k_B T}{m \omega_a \omega_0} e^{\beta V_b}
\end{align} 
where $\omega_0 = m^{-1}U''(0)$ and $\omega_a=m^{-1}U''(a)$ are constants that depend on the second derivative of $U(x)$ at the origin and $a$ respectively. Comparing the above equation with \eqref{forwardrate}, we get:
\begin{align}
    \frac{E_a L(\tau_1,a)}{E_a T_b} \sim D\frac{m \omega_a \omega_0}{2\pi k_B T}
\end{align}
as $E_a L(\tau_1,a) \sim e^{\beta V_b}$. 

\section{Conclusion}
The factorization of the reaction rate into an excursion measure term and a local time fluctuation term, as in the RHS of \eqref{kappa} and \eqref{forwardrate}, allows a direct comparison with the Kramers (and TST) rate expression in the Smoluchowski (high friction/overdamped) limit.  We emphasize that the factorization of the rate expression is expected to hold even in cases where an explicit transition state cannot be identified in the PES, the type of problems for which TPS and TPT are useful; the existence of the limits in \eqref{kappa} is only conditional upon $(X_t)$ being positive recurrent and regenerative, which is a reasonable assumption for systems with discernible `reactants' and `products' metastable states. Since the expression for reaction rate remains structurally the same for all elementary chemical reaction, we expect this factorization to hold even for general multidimensional potential energy surfaces. Hence we believe that further work using an excursion-theoretic interpretation of chemical reactions can potentially unravel the relevant fluctuation timescales intrinsic to the underlying PES.

\begin{acknowledgments}
We thank Stuart Althorpe for reading through the introduction of the manuscript and providing useful suggestions. V.G.S.\ acknowledges support from St.\ John's College, University of Cambridge, through a Dr. Manmohan Singh Scholarship. Both V.G.S. and R.B. acknowledge funding from a SERB matrix grant (\#MTR/2017/000750).
\end{acknowledgments}

\section*{Conflict of interest}
The authors have no conflicts to disclose.

\section*{Data Availability}

Data sharing is not applicable to this article as no new data were created or analyzed in this study.

\selectlanguage{english}
\bibliographystyle{unsrt}
\bibliography{Citations}
\end{document}